\documentclass[11pt]{article}
\usepackage[latin1]{inputenc}
\usepackage{amsmath, graphicx, layout, verbatim}
\usepackage{setspace}
\usepackage{fancyhdr}
\usepackage{float}
\usepackage{color}
\usepackage{natbib}
\usepackage{verbatim}
\usepackage{enumitem}
\usepackage{amsthm}
\usepackage{amsfonts}
\usepackage{amssymb}
\usepackage[pagebackref,colorlinks=true,urlcolor=blue,citecolor=blue,linkcolor=blue,bookmarks=false]{hyperref}

\def\I{{\mathbb I}}

\newtheorem{thm}{Theorem}

\newtheorem{Lemma}{Lemma}
\newtheorem{Assumption}{Assumption}

\title{Learning with many experts:  model selection and sparsity}
\author{
Rafael Izbicki and Rafael Bassi Stern \footnote{Department of Statistics, Carnegie Mellon University}}
\date{\today}

\newcounter{tableNotation}
\newcounter{SModels}
\newcounter{simulated}
\newcounter{wine}
\newcounter{ion}
\newcounter{astroLambda}
\begin{document}
%\linenumbers
\maketitle 

\emph{This is the pre-peer reviewed version of the following article: Izbicki, R., Stern, R. B. `` Learning with many experts: Model selection and sparsity." Statistical Analysis and Data Mining 6.6 (2013): 565-577., which has been 
published in final form at \url{http://onlinelibrary.wiley.com/doi/10.1002/sam.11206/full}}

\abstract{
Experts classifying data are often imprecise. Recently, several models 
have been proposed to train classifiers using the noisy labels generated 
by these experts. How to choose between 
these models? In such situations, the true labels are 
unavailable. Thus, one cannot perform model selection using the standard versions of methods such as empirical risk 
minimization and cross validation. In order to allow model selection, we present a surrogate loss and 
provide theoretical guarantees that assure its consistency.  Next, we discuss how this loss can be used to tune a penalization which introduces sparsity in 
the parameters of a traditional class of models. Sparsity provides more parsimonious models and can avoid 
overfitting. Nevertheless, it has seldom been discussed in the context of noisy labels due to the difficulty in 
model selection and, therefore, in choosing tuning parameters. We apply these techniques to several sets of simulated and real data.
}

\section{Introduction}

In many situations, getting reliable labels in a dataset is very expensive and therefore assigning highly trained experts to do such tasks is undesirable. In other situations, even trained experts disagree about the labels of the data. Cases like these include spam detection, diagnosis 
of patients based on images and morphological classification of galaxies \cite{galaxyZoo}. Although it might be expensive to train experts to a degree one can trust their labels \cite{joey}, systems such as Amazon Mechanical Turk \cite{schulze:exploring} allow each sample unit to be classified by many (not necessarily perfect) experts by a reasonably small cost. These experts do not have to be people. For instance, they can be different cheap screening tests in a medical problem \cite{diagnostic}. In these situations, it is desirable to have methods that can detect how reliable each expert is and use this information not only to 
detect the adequate labels of the data but also to train accurate classifiers to predict new data \cite{costSensitive}. These methodologies are usually called crowdsourcing methods. Here we focus on predicting binary variables, even though similar ideas can be used in the case of predicting a categorical variable with more than two labels.

The most common approach to deal with multiple experts is to first consider a majority vote scheme to input the labels of each sample unit. Such procedure is known to be suboptimal in many situations \citep{Yan2010}. Many other approaches have been proposed recently. While some of them are based on a two step procedure of first trying to find the true labels in the data and then training classifiers based on them \cite{Chittaranjan, crowd} others do these tasks simultaneously, that is, the classifier is trained by assuming that the labels from the experts may be incorrect \cite{Raykar2010, Yan2010}. We follow the latter approach, even though the model selection technique we propose works for the former method as well.

Also, many of the existing methods are essentially algorithm-based \cite{Donmez}. However, a significant amount of the recent methods 
consist of probabilistic approaches to this problem, in which the unobserved true label is modeled as 
a latent variable \cite{Welinder, Provost, Raykar2010, Yan2010,Yan10,AAAIW125257}. In the latter case, the parameters of 
the model are usually estimated through the Expectation Maximization Algorithm \cite{EM}. This approach has roots on \citep{dawid}. However, 
less emphasis has been given to develop ways of comparing these different models. Since the  usual techniques for model selection depend on observing 
the real labels of the data, they cannot be used in this case.  \cite{Lam:Stork:03} discusses how to find good models when only one 
annotator is available . Here we extend some of these results and relax some of the assumptions made.
We take a predictive approach: by good models
we mean models that have low predictive errors. 

The literature also lacks on methods that can build sparse (in terms of coefficients of the model related to the features) classifiers in crowdsourcing methods. Sparsity is a useful tool when trying to build classifiers that have good generalization properties, that is, that do not suffer from 
overfitting. Moreover, many common models used for crowdsourcing have a number of parameters that grow with both the number of experts and samples. Having too many parameters can increase the prediction error substantially. Introducing sparsity on such classifiers leads to more parsimonious models that potentially have better performance.
Bayesian methods such as the one used by \cite{Raykar2010} can lead to shrinkage of the coefficients and therefore to better prediction errors, however 
it is not clear how to choose prior hyperparameters on them when one aims at good prediction errors. Sparse methods are also valuable
because they can reduce costs: for example, in new samples a smaller number of variables have to be measured.

In Section \ref{modelSelection} we develop a method for model 
selection. In Section \ref{model} we present a model which allows sparse solutions. We also show how to fit the model parameters for a fixed value of the parameter which specifies the amount of sparsity. Section \ref{applications} provides applications
of both the model selection technique and the sparse model we propose.
In particular, we use our model selection technique to select the tuning parameter which induces the classifier with best predictive errors

\section{Model Selection}
\label{modelSelection}

Assume there are $d$ experts that label $n$ sample units. For each of these units, we measure $k$ features. $X_{ij}$ is the $j$-th feature of the $i$-th unit. 
$X_{i}$ denotes the vector of all features for the $i$-th sample. $Y_{ir}\in \{0,1\}$ denotes the label attributed by the $r$-th expert (annotator) to the $i$-th 
sample. $Z_i \in \{0,1\}$ is the unobserved variable which corresponds to the appropriate label for the $i$-th sample unit. 
Table \arabic{tableNotation} contains a summary of this notation.

 \begin{table}[!ht]
 \caption{Model's random variables}
 \centering
 \begin{tabular}{|c|ccc|ccc|} \hline
 		True Labels 		& \multicolumn{3}{|c|}{Experts' Labels}& \multicolumn{3}{c|}{Features} \\ \hline \hline
 		$Z_1$ 					& $Y_{11}$ & \ldots & $Y_{1d}$ 				 & $X_{11}$ & \ldots & $X_{1k}$ \\ 
 		\vdots &\vdots  &	$\ddots$ &\vdots 										 &\vdots &$\ddots$ & \vdots \\ 
 		$Z_n$ 					& $Y_{n1}$ & \ldots & $Y_{nd}$				 & $X_{n1}$ & \ldots & $X_{nk}$ \\ \hline
 \end{tabular}
 \label{tableNotation}
 \end{table}

We assume one wishes to find a classifier which minimizes the 0-1 loss. That is, one is interested in finding a classifier that has small probability of making a mistake on a new sample. In this 
case, many techniques of model selection rely on calculating empirical errors on a test data set \cite{Hast:Tibs:Frie:2001}. When using noisy labels, the empirical error is unavailable and this strategy cannot be directly applied. In order to overcome this difficulty, we introduce a score which is closely related to the empirical error. The reliability of this score does not depend on assuming that the data is generated according to the model in Section \ref{model}. 

Our score is based on splitting the data into a training set and testing set, $(X^{test}_{i},Y^{test}_{i})_{1 \leq i \leq n'}$. If this cannot be done due to a small sample size, one can use a cross validated version of it \cite{Hast:Tibs:Frie:2001}. 
Consider a set of models, $\Lambda$. For example, this set can be composed of all models generated by different $\lambda$ values in the model presented in Section \ref{model}. It could also be the set of models fit with different subsets of the features or it could even contain different models such as those obtained using majority vote to input the labels or models such as in \cite{Raykar2010}. 

For each $\lambda \in \Lambda$, we train the model using the training set. Call $z^{\lambda}$ the classifier $\lambda$ obtained from the training data. Through model selection we wish to find $\lambda \in \Lambda$ with the smallest risk. Define the risk of $\lambda$ as $R(\lambda)=E[I(Z\neq z^{\lambda}(X))]$, that is, the probability of a new sample unit being misclassified by the classifier $\lambda$. Let $n$ be the sample size in the testing data set. We use $z_{i}^{\lambda}$ as shorthand for $z^{\lambda}(X_{i}^{test})$. 
The (in practice incalculable) empirical risk of model $\lambda$ is $\hat{R}(\lambda) = \frac{1}{n'}\sum_{i=1}^{n'}{\I(z^{\lambda}_i \neq Z_{i} )}$. 
For each $\lambda \in \Lambda$, we score how bad $z^{\lambda}$ performs through $\hat{S}$,
\[ \hat{S}(\lambda) =\frac{1}{n'} \sum_{i=1}^{n'}{\frac{1}{d}\sum_{j=1}^{d}{I(z_{i}^{\lambda} \neq Y^{test}_{i,j})}}, \]

and select the model $\lambda^{*}$ such that
$$ \lambda^{*} = \arg\min_{\lambda \in \Lambda} \hat{S}(\lambda) $$

We prove $\lambda^{*}$ is consistent (in the sense of asymptotically giving the same results as when minimizing the real risk
$R(\lambda)$) and provide an upper bound on its rate of convergence. In the following theorems, $VC(\Lambda)$ is the VC-dimension of $\Lambda$ and $D$ a universal constant defined in \cite{Vaart}.

\begin{Assumption}
\label{assumptionInd}
For every $i \neq j$, $(X_{i},Z_{i},(Y_{i,k})_{k=1}^{d})$ is independent of $(X_{j},Z_{j},(Y_{j,k})_{k=1}^{d})$.
\end{Assumption}

\begin{Assumption}
\label{assumptionError}
Let $\epsilon_{j}=P(Z \neq Y_{j})$ be the imprecision of expert $j$. $\bar{\epsilon}\equiv\frac{\sum_{j=1}^{d}{\epsilon_{j}}}{d} < \frac{1}{2}$.
\end{Assumption}

Assumption \ref{assumptionError} means that the label provided by an expert picked uniformly is better than the flip of a coin. We now consider two additional assumptions
and then prove that using $\hat{S}$ to perform model selection works under either of them.

\begin{Assumption}
\label{assumptionIndModelExpert}
For all $i$, $Cov\left(\frac{\sum_{j=1}^{d}{\left(1-2I(Z_{i} \neq Y^{test}_{i,j})\right)}}{d},I(z^\lambda_{i} \neq Z_{i})\right) = 0$.
\end{Assumption}

Assumption \ref{assumptionIndModelExpert} holds e.g. when the classifier and expert errors are unrelated, a condition that appears 
on \cite{Lam:Stork:03}. 

\begin{Assumption}
\label{assumptionIndExpertExpert}
For all $i$ and $j \neq j^{*}$, $Cov(I(Z_{i} \neq Y^{test}_{i,j}),I(Z_{i} \neq Y^{test}_{i,j^{*}})) = 0$.
\end{Assumption}

Assumption \ref{assumptionIndExpertExpert} holds e.g. when the errors of every two experts are unrelated.
 We prove the following Theorems:

\begin{thm}
	\label{modelSelectionMain}
	Under Assumptions \ref{assumptionInd} and \ref{assumptionIndModelExpert}, if $\Lambda$ is a VC-Class,

	\[ P(\sup_{\lambda \in \Lambda}|\hat{S}(\lambda)-(1-2\bar{\epsilon})R(\lambda)-\bar{\epsilon}| > \delta) \leq \left(\frac{D\sqrt{n'}\delta}{\sqrt{2VC(\Lambda)}}\right)^{2 VC(\Lambda)} e^{-2n'\delta^{2}}. \]
\end{thm}

\begin{thm}
	\label{modelSelectionMain2}
	Under Assumptions \ref{assumptionInd} and \ref{assumptionIndExpertExpert}, if $\Lambda$ is a VC-Class,
	
	\[ P(\sup_{\lambda \in \Lambda}|\hat{S}(\lambda)-(1-2\bar{\epsilon})R(\lambda)-\bar{\epsilon}| > \frac{1}{4\sqrt{d}}+\delta) \leq \left(\frac{D\sqrt{n'}\delta}{\sqrt{2VC(\Lambda)}}\right)^{2 VC(\Lambda)} e^{-2n'\delta^{2}}. \]
\end{thm}

The proofs of these facts are sketched in Appendix \ref{appendixProof}. 
Thus, Theorem \ref{modelSelectionMain} states that, under assumptions \ref{assumptionInd} and \ref{assumptionIndModelExpert}, as $n$ increases, with 
high probability, $\hat{S}(\lambda)$ will not deviate more than roughly $\frac{1}{\sqrt{n'}}$ 
from $(1-2\bar{\epsilon})R(\lambda)+\bar{\epsilon}$ (uniformly). Moreover, under assumption \ref{assumptionError}, $(1-2\bar{\epsilon})R(\lambda)+\bar{\epsilon}$ 
increases on $R(\lambda)$. Hence, the minimizer of $\hat{S}(\lambda)$ will be close to the minimizer of 
$R(\lambda)$. Using Theorem \ref{modelSelectionMain2}, the same type of reasoning 
applies under \ref{assumptionIndExpertExpert}, with the exception that $\hat{S}(\lambda)$ will not deviate more 
than roughly $\frac{1}{\sqrt{n'}}+\frac{1}{4\sqrt{d}}$ from $(1-2\bar{\epsilon})R(\lambda)+\bar{\epsilon}$. Hence, consistency is obtained 
only if the number of experts also increases. Next section describes how to introduce sparsity on a particular model from the literature. In Section \ref{applications}
we discuss how to select the tuning parameter for this model using $\widehat(S)$.

\section{Model Description and Sparse Fitting}
\label{model}

The model we use is described by the following conditions, where we use the same notation as in Table \arabic{tableNotation}:

\begin{enumerate}
	\item $(Z_{i})_{i \leq n}$ are conditionally independent given $(X_{i})_{i \leq n}$.
	\item $Z_i|X_{i}=x_{i} \sim Ber\left(\frac{\exp{\{\beta_0+\sum_{j=i}^k \beta_j x_{ij}}\}}{1+\exp{\{\beta_0+ \sum_{j=i}^k \beta_j x_{ij}}\}}\right).$
	\item $(Y_{ij})_{i \leq n, j \leq k}$ are conditionally independent given $(Z_{i})_{i \leq n}$ and $(X_{i})_{i \leq n}$.
	\item $P(Y_{ik} \neq Z_{i}|Z_{i}=z_{i},X_{i}=x_{i}) = \frac{1}{1+exp\{\alpha_k+\sum_{j=i}^k\gamma_j x_{ij}\}}$.
\end{enumerate}

This model is similar to the one specified in \cite{Yan2010} with the exception that the $\gamma_j$ coefficients do not depend on the expert. The model's parameters can be interpreted. The higher an $\alpha_k$ is, the more precise the $k$-the expert. Also, the $\gamma_j$ coefficients explain how each feature influences the difficulty in classifying a sample unit. $\beta'$s
are the coefficients that measure the influence of the covariates on the real response. One implicit assumption is that the influence of each feature is the same for all experts. 
The number of parameters in this model is more than twice the number of features. Thus, sparse classifiers might improve the prediction error if $n$ is small \cite{Hast:Tibs:Frie:2001}.

We choose this model because it is simple enough and yet sufficiently reasonable to be applied to many practical situations. We do not intend
to argue that this is the best model in all situations. However, similar ideas of how to introduce sparsity can be used in other models from the literature.

The joint distribution of $(Y,Z)$ given $X$ corresponds to a mixture of products of independent Bernoulli variables. In fact, denoting

$$\mu_i:=\frac{\exp{\{\beta_0+\sum_{j=i}^J \beta_j x_{ij}}\}}{1+\exp{\{\beta_0+ \sum_{j=i}^J \beta_j x_{ij}}\}},$$

the complete likelihood (conditional on the features) is given by
\begin{align*}
 &L(y,z;\theta,x) = \\[2mm]
 &\prod_i P(\forall k, Y_{ik}=y_{ik},Z_i=z_i|X_{i}=x_{i}) =  \prod_i  P(\forall k,Y_{ik}=y_{ik}|Z_i=z_i,X_{i}=x_{i})P(Z_i=z_i|X_{i}=x_{i}) =  \\
 &\prod_i \left( \prod_{k } P(Y_{ik}=y_{ik}|Z_i=z_i,X_{i}=x_{i}) \right) P(Z_i=z_i|X_{i}=x_{i})= \prod_i \mu_i^{z_i} (1-\mu_i)^{1-z_i} \times b_i,
\end{align*}

 where 
 
 $$b_i=  \prod_{k } \left[\left(\frac{\exp\{\alpha_k+\sum_j\gamma_j x_{ij}\}}{1+\exp\{\alpha_k+\sum_j\gamma_j x_{ij}\}}\right)^{y_{ik}} \left(\frac{1}{1+\exp\{\alpha_k+\sum_j\gamma_j x_{ij}\}}\right)^{1-y_{ik}} \right]^{z_i}\times$$
 
 $$\times
\left[\left(\frac{1}{1+\exp\{\alpha_k+\sum_j\gamma_j x_{ij}\}}\right)^{y_{ik}} \left(\frac{\exp\{\alpha_k+\sum_j\gamma_j x_{ij}\}}{1+\exp\{\alpha_k+\sum_j\gamma_j x_{ij}\}}\right)^{1-y_{ik}} \right]^{1-z_i}  $$

(that is, $b_i$ is the joint probability of the experts responses conditional on the true labels and on the explanatory variables) and $\theta$ indicates all of the model's parameters. Hence, the (complete) log-likelihood is given by
\begin{align}
\label{logLikeli}
&l(y,z;\theta,x)= \\[2mm]
&\sum_i \left( \sum_{k } d_{ik}\log\left(\frac{\exp\{\alpha_k+\sum_j\gamma_j x_{ij}\}}{1+\exp\{\alpha_k+\sum_j\gamma_j x_{ij}\}}\right)+
(1-d_{ik})\log\left(\frac{1}{1+\exp\{\alpha_k+\sum_j\gamma_j x_{ij}\}}\right)\right) + \notag\\
&+z_i \log(\mu_i)+(1-z_i) \log(1-\mu_i),\notag 
\end{align}

where
$$ d_{ik}:=1+2z_iy_{ik}-z_i-y_{ik}.$$

Traditionally, the (local) maximum of the marginal likelihood (defined as $L(y;\theta,x)=\sum_z L(y,z;\theta,x)$) is found by using the EM algorithm \cite{EMOriginal}.
We propose to introduce sparsity to this solution. Sparsity reduces the 
number of parameters we have to estimate, and hence can improve the prediction error. For a comprehensive account of the role of 
sparsity on prediction problems, the reader
is referred to \citep{Hast:Tibs:Frie:2001}. To find a sparse fit of the model, instead of maximizing the marginal likelihood, we introduce
a $L_1$-penalty in the function. That is, we compute,

\begin{equation}
\label{lasso}
	\arg \sup_{\theta} \mbox{ } \left( l(y;\theta) - \lambda\sum_{j=1}^k|\beta_j|- \lambda\sum_{j=1}^k|\gamma_j| \right),
\end{equation}

for some fixed $\lambda>0$. $l(y;\theta)$ is the log-likelihood of the observed noisy labels, $y$. Section 
\ref{modelSelection} indicates how one can pick an optimum value of $\lambda>0$. Other penalties (e.g., $L_2$)
could also lead to better prediction errors, however $L_1$ penalty creates sparse solutions (that is, it not only shrinks the coefficients)
and, as we will see, is tractable from a computational point of view. In order to solve Equation \ref{lasso}, we will first 
rephrase it in terms of a Bayesian problem that leads to the same results. Imagine that we assign a prior probability for $\theta$ as follows: 

\begin{equation}
	\label{lassoPrior}
	\pi(\theta) \propto \exp\left(-\lambda\sum_{j=1}^k|\beta_j|-\lambda\sum_{j=1}^k|\gamma_j|\right).
\end{equation}

The maximum a posteriori estimate (MAP) for $\theta$, given $Y$ and $X$, corresponds to the solution of Equation \ref{lasso}. Let

$$ g(\theta,z) := l(y,z;\theta,x) - \lambda\sum_{j=1}^k|\beta_j|- \lambda\sum_{j=1}^k|\gamma_j|, \mbox{ }$$

where $l(y,z;\theta,x)$ is as in Equation \ref{logLikeli}.
To 
find the MAP estimate we use a MAP-EM algorithm \cite{EM}. 
That is, we first initialize $\theta$ with some given values. Then,  we iterate until convergence:

\begin{enumerate}
	\item \textbf{(Expectation step)} Find the expected value of the $g(\theta,Z)$, conditional on the current estimates of the parameters $\theta$ and on $y_{ij}$ (denoted by $E[g(\theta,Z)]$).
	\item \textbf{(Maximization step)} Maximize $E[g(\theta,Z)]$ with respect to $\theta$.
\end{enumerate}

Since $g(\theta,Z)$ is linear in $Z$, the Expectation step follows directly from calculating
\[E[Z_i|Y_{ik}=y_{ik}\mbox{ }\forall i,k]=	\frac{\mu_i \cdot \exp\{\sum_{k} y_{ik}*(\alpha_k+\sum_j\gamma_j x_{ij})\}}{\mu_i \cdot \exp\{\sum_{k} y_{ik}*(\alpha_k+\sum_j\gamma_j x_{ij})\}+(1-\mu_i)\exp\{\sum_{k} (1-y_{ik})*(\alpha_k+\sum_j\gamma_j x_{ij})\}}. \]

and plugging these values into $g(\theta,Z)$. Denote by $g(\theta,z)$ the expected value of $g(\theta,Z)$.
For the Maximization step, observe that 
\begin{subequations}
\label{eqMax}
\begin{align}
 &\arg \sup_{\theta} \mbox{ } g(\theta,z)= \label{eqBla2}\\[1.5mm]
 &\arg \sup_{\gamma's,\alpha's} \mbox{ } \sum_i \left( \sum_{k \in A_i} d_{ik}\log\left(\frac{\exp\{\alpha_k+\sum_j\gamma_j x_{ij}\}}{1+\exp\{\alpha_k+\sum_j\gamma_j x_{ij}\}}\right)+
(1-d_{ik})\log\left(\frac{1}{1+\exp\{\alpha_k+\sum_j\gamma_j x_{ij}\}}\right)\right)  \label{eq1} \\
&- \lambda\sum_{j=1}^k|\gamma_j| \label{eqbla} \\[1.5mm] 
&+ \arg \sup_\beta \mbox{ } \sum_i \left(z_i \log(\mu_i)+(1-z_i) \log(1-\mu_i) \right) - \lambda\sum_{j=1}^k|\beta_j|.  \label{eq2}
\end{align}
\end{subequations}

Hence, we have two independent maximization problems, \ref{eq1} and \ref{eq2}. Each of them correspond to solving for Weighted L1-Regularized Logistic Regressions, which is implemented in functions such as glmnet \cite{glmnet} in R.
More details on this are given in Appendix \ref{appendixMax}.

The MAP-EM often converges to different points according to the initialization values. One reason for this is that them MAP-EM is
guaranteed to converge only to local maximums. A more important reason is due to a type of non-identifiability \cite{Reilink} in the model. The parameters $(\alpha,\gamma,\beta)$ and $(-\alpha,-\gamma,-\beta)$ induce the same distribution for the data\footnote{For example, consider there is only one expert and that $Z$ represents if a patient is sick or not. We get the same probability that the expert finds the patient to be sick when the expert has good accuracy and the patient has a high probability of being sick (parameters $(\alpha,\gamma,\beta)$) and when the expert has a bad accuracy and the patient has a small probability of being sick (parameters $(-\alpha,-\gamma,-\beta)$).}. This is common in mixture models and is known as trivial non-identifiability \cite{practicalIdent}. Consequently, the likelihood will have two optimizers. In order to choose between these points we assume that, averaging over all experts, the probability of correct classification is larger 
than $50\%$. This assumption was discussed in Section \ref{modelSelection} and can also be found in \cite{crowd}. Using this assumption, if the MAP-EM converges to $\theta$, we choose between $\theta$ and $-\theta$, selecting the classifier which agrees the most with majority vote.

%The MAP-EM depends on the value of $\lambda$ specified on the prior for the model parameters. We treat $\lambda$ as a tuning parameter. A method for choosing the value of $\lambda$ which minimizes the prediction error is described in next section.
Next section shows empirical performance of this method and the model selection technique in both simulated and real datasets. In particular we discuss how 
to use the model selection technique from Section \ref{modelSelection} to choose the tuning parameter $\lambda$.

\section{Experiments}
\label{applications}

We perform 4 experiments that aim at exploring the two methods proposed (sparsity and model selection). 
Experiment in \ref{simulatedData}, is completely simulated: we generate the features, real responses and also responses from experts. This allows the Bayes
error to be calculated. Experiments \ref{ionData} and \ref{wineData} use data from the UCI repository \cite{Newman+Hettich+Blake+Merz:1998}. These databases only contain features and appropriate labels and, thus, we complement them with simulated responses from hypothetical experts. In \ref{simulatedData}, \ref{ionData} and \ref{wineData} we generate the votes from the experts in three ways: 

\begin{enumerate}
\item The probabilities of misclassification do not depend on the observed features,
\item The probabilities of misclassification follow the model described in Section \ref{model}
\item The probabilities of misclassification do not follow Section \ref{model}.
\end{enumerate}

The exact description of how the votes were generated varies and is described in each example.

Experiment \ref{astroData} presents a real data set in which a large set of experts responses (42) is available. Hence, majority vote gives us the real response with high probability. For instance, assuming each expert is correct with probability 70\% and the responses from experts are independent, majority vote would get the right label with probability 
$\approx$ 99.5\%. In this example, (i), (ii) and (iii) correspond to taking random subsets of size 3 from the 42 experts and comparing the results we get with the (reliable) majority vote on the 42 experts, as if these were the true labels.

In each experiment, we fit and compare the EM without sparsity (denoted by \emph{EM}), with sparsity (\emph{EM-Sparse}) and a L1-penalized logistic regression on the labels obtained by majority vote (\emph{Majority}). For each of the classifiers obtained, we compute $\hat{S}$ and compare it to $\hat{R}$ (which in practice 
would not be available), the empirical risk. 
For the sake of comparison, we also fit a L1-penalized logistic regression on the real labels. 

We initialize all the parameters generating Gaussian variables with variance $1$. For the $\alpha$'s and $\gamma$'s we pick mean $0$. For the $\beta$'s, the mean is the corresponding coefficient of the logistic regression fitted through majority vote. 
In order to avoid local maximums, this procedure was repeated 30 times for each simulation.

\subsection{Simulated Data Set}
\label{simulatedData}

We take sample size $2500$. The logit of the probability of each appropriate label being $1$ is $\beta_0+\sum_{j=i}^5 \beta_j x_{ij}$  with $\beta=(-0.1,1,0.25,0.24,-0.3,-0.2)$. $(X_1,X_2,X_3,X_4,X_5)$ follows a multivariate normal with mean mean $(1, 2, 3, 4, 5)$ and covariance matrix,

 \[ \left( \begin{array}{ccccc}
0.50 &0.10 &0.25 &0.10 &0.10 \\
0.10 & 0.50 &0.10 &0.05 &0.04  \\
0.25 & 0.10 & 0.80 &0.01 &0.10 \\
0.10 & 0.05 & 0.01& 0.40 & 0.10 \\
0.10 & 0.04 & 0.10 &0.10  & 0.50  \end{array} \right)\] 

We also include 50 covariates unrelated to the labels generated independently from a standard normal distribution. We generate the experts' responses in the following ways:

\begin{enumerate}[label=(\roman{*}), ref=(\roman{*})]
	\item Three experts with misclassification probabilities 0.5, 0.15 and 0.47.
	\item Four experts, with misclassification probabilities as in Section \ref{model}, with $\alpha=(0,.75,-.1)$ and 
	
	$\gamma=(.1,.2,-.08,.025,-.065)$.
	\item Three experts, with probabilities as in Section \ref{model}, with $\alpha=(0,.65,-.12)$ and $\gamma=(.05,.05,-.1,-.1,0)$ but generating the votes through the square of the covariates. 
\end{enumerate}

Figure \arabic{simulated} shows the results of applying the model selection ideas to tune the parameter $\lambda$. 
It also shows the estimated predictive risk (based on the real labels) for
\emph{EM}, \emph{EM Sparse} and \emph{Majority}, with an interval with one standard 
deviation around the mean.  The Bayes risk is represented by a horizontal line. It is possible to see that in (i), (ii) and (iii), \emph{EM Sparse} 
beats the other models. Moreover, plain \emph{EM} does not give satisfactory results. This 
is because there are many (noninformative) covariates, and hence introducing sparsity becomes crucial. Figures related to (ii) and (iii) also 
show that $\widehat{S}$ is also a useful tool to detect points in which either the EM did not 
converge: they are the points that have a very different behavior in these curves.

Finally, the results from Table \arabic{SModels} agree with our analysis: using $\widehat{S}$ to select among different methods gives the same results as 
using $\widehat{R}$, that is, when using $\widehat{S}$ we also conclude that \emph{EM Sparse} is the best model in this case.

\subsection{Ionosphere Data Set}
\label{ionData}

The data set ion holds 351 radar returns which can be ``good'' or ``bad''. There are 34 continuous features. We simulate the expert labels, using at most the $4$ first features, in the following ways:

\begin{enumerate}[label=(\roman{*}), ref=(\roman{*})]
	\item Five experts with misclassification probabilities 0.6, 0.2, 0.5, 0.4 and 0.4.
	\item Four experts, with misclassification probabilities as in Section \ref{model}, with $\alpha=(0.7,0,1.6,0.7)$ and 
	
	$\gamma=(0.3,0.25,-0.3,0.1)$. 
	\item Four experts, with probabilities as in Section \ref{model}, $\alpha=(0.7,0,1.6,0.7)$ and $\gamma=(0.3,0.25,-0.3,0.1)$, but generating the votes through the square of the covariates.  
\end{enumerate}

We use a training set of size $175$. Figure \arabic{ion} shows how we fitted \emph{EM-Sparse} and compares it to the other models. $\hat{S}$ is approximately monotonically increasing with $\hat{R}$ and, thus, the minimizer of $\hat{S}$ has empirical risk close to that of the empirical risk minimizer. In scenario (ii), although $\lambda^{*}$ is far from the one which 
minimizes the empirical risk, their risk is similar. Abrupt variations in the top graphs also indicate cases in which the EM probably did not converge. On the bottom, \emph{EM-sparse} improves on results
of both \emph{EM} and \emph{Majority} in all scenarios.
Finally, we see from the results of Table \arabic{SModels} that using $\widehat{S}$ to select between the different models indicates that \emph{EM Sparse} is the model with smaller estimated predictive risk $\hat{R}$ on these cases.

\subsection{Wine Quality Data Set}
\label{wineData}

The data set wine contains $1599$ red wines and $11$ features such as alcohol content and pH. The wine quality of a sample unit is a number between $0$ and $10$. We define the appropriate label as $1$ if wine quality is greater than $5$ and $0$, otherwise. We generate the noisy labels, using at most the $5$ first features, in the following way:

\begin{enumerate}[label=(\roman{*}), ref=(\roman{*})]
	\item Three experts with misclassification probabilities 0.4, 0.3 and 0.5.
	\item Four experts, with misclassification probabilities as in Section \ref{model}, with $\alpha=(1,-0.5,2.1,2.3)$ and 
	
	$\gamma=(0.25,0.4,0.3,0)$.
	\item Three experts, with probabilities as in Section \ref{model}, $\alpha=(1,-0.5,1.2)$ and $\gamma=(0.1,0.2,-0.2,-0.3,-0.3)$, but generating the votes through the square of the covariates.
\end{enumerate}

We use a training set of size $1000$. Figure \arabic{wine} shows how we fitted \emph{EM-Sparse} and compares it to the other models. Regarding the bottom of the figure, in (i) sparsity reduces the prediction error: both \emph{EM-Sparse} and \emph{Majority} are as good as the model fitted using the real labels and much better than \emph{EM}. In (ii), \emph{Majority} is worse than the other approaches, which have the same performance. In (iii), all models perform close to the one obtained using the real labels. On the top of (iii), $\lambda^{*}$ is far from the one which minimizes $\hat{R}$, but has approximately the same risk.
Notice that Table \arabic{SModels} leads us to similar conclusions, hence using model selection ideas introduced here also helps us to decide on what is the best approach, EM or majority vote.

\subsection{Astronomy Data Set}
\label{astroData}

The sample units in this data set are galaxies. The label is $1$ if the shape of the galaxy is \emph{regular} \cite{izbicki} and $0$, otherwise. Each galaxy has been labeled by $42$ astronomers from CANDELS team \cite{citeulike:9923967}. For each galaxy, there are 7 features which are summary statistics of the their images. These statistics are further described in \cite{izbicki} and \cite{lotz}. The training set is composed of $90$ galaxies and the testing set of $85$. We perform three experiments, (i), (ii) and (iii), by picking as the noisy labels random subsets of size $3$ out of the $42$ astronomers. 
True labels are defined to be the majority vote over the 42 astronomers.

Figure \arabic{astroLambda} illustrates the procedure of fitting \emph{EM-Sparse} and compares it to \emph{EM} and \emph{Majority}. On the top, minimizing $\hat{S}$ yields the same result as minimizing $\hat{R}$. On the bottom, \emph{EM-Sparse} and \emph{Majority} have approximately the same performance, close to the performance of the model that was fitted when using the real labels. 
On the other hand, using \emph{EM} without introducing sparsity leads to slightly worse prediction errors in (iii).
We emphasize that the large confidence intervals are due to a small sample size. Hence, it is difficult to get conclusive results of which model is
the best in this case. However, the first row of Figure \arabic{astroLambda} shows in practice that assumptions made in Section \ref{modelSelection} for model selection are reasonable for this problem.

 \begin{figure}[htpb!]
   \centering
   \includegraphics[scale=1]{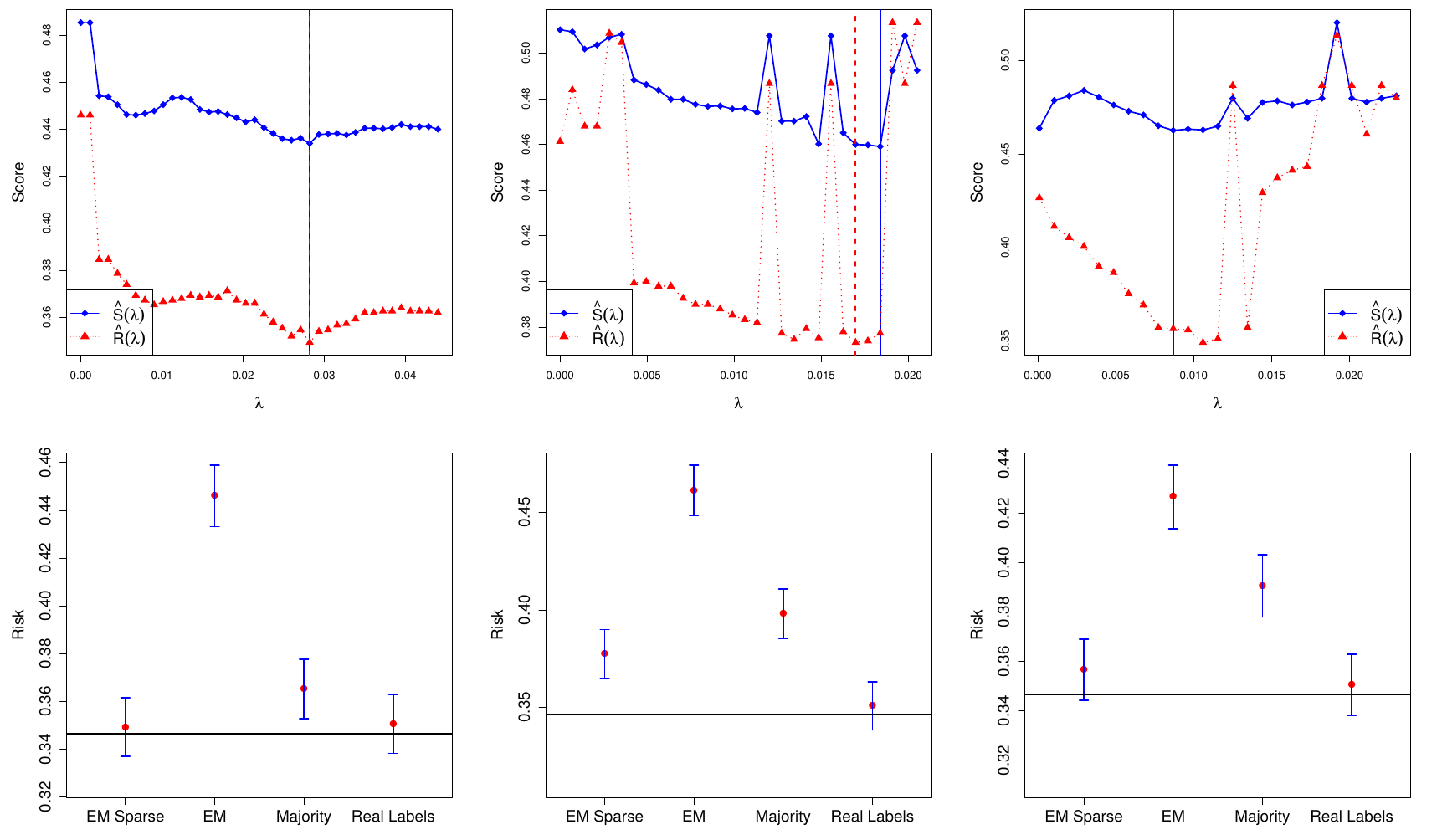}
   \caption{Top: Process of choosing $\lambda$ using $\widehat{S}(\lambda)$ for Simulated Data Set for models (i), (ii) and (iii) respectively. Vertical lines indicate where the minimum is attained for $\widehat{S}(\lambda)$ (solid) and for $\widehat{R}(\lambda)$ (dashed). Bottom: Estimated prediction errors for each dataset according to each model. Horizontal lines indicate error of the Bayes classifier.}
  \label{simulated}
\end{figure}

\begin{figure}[htpb!]
  \centering
  \includegraphics[scale=1]{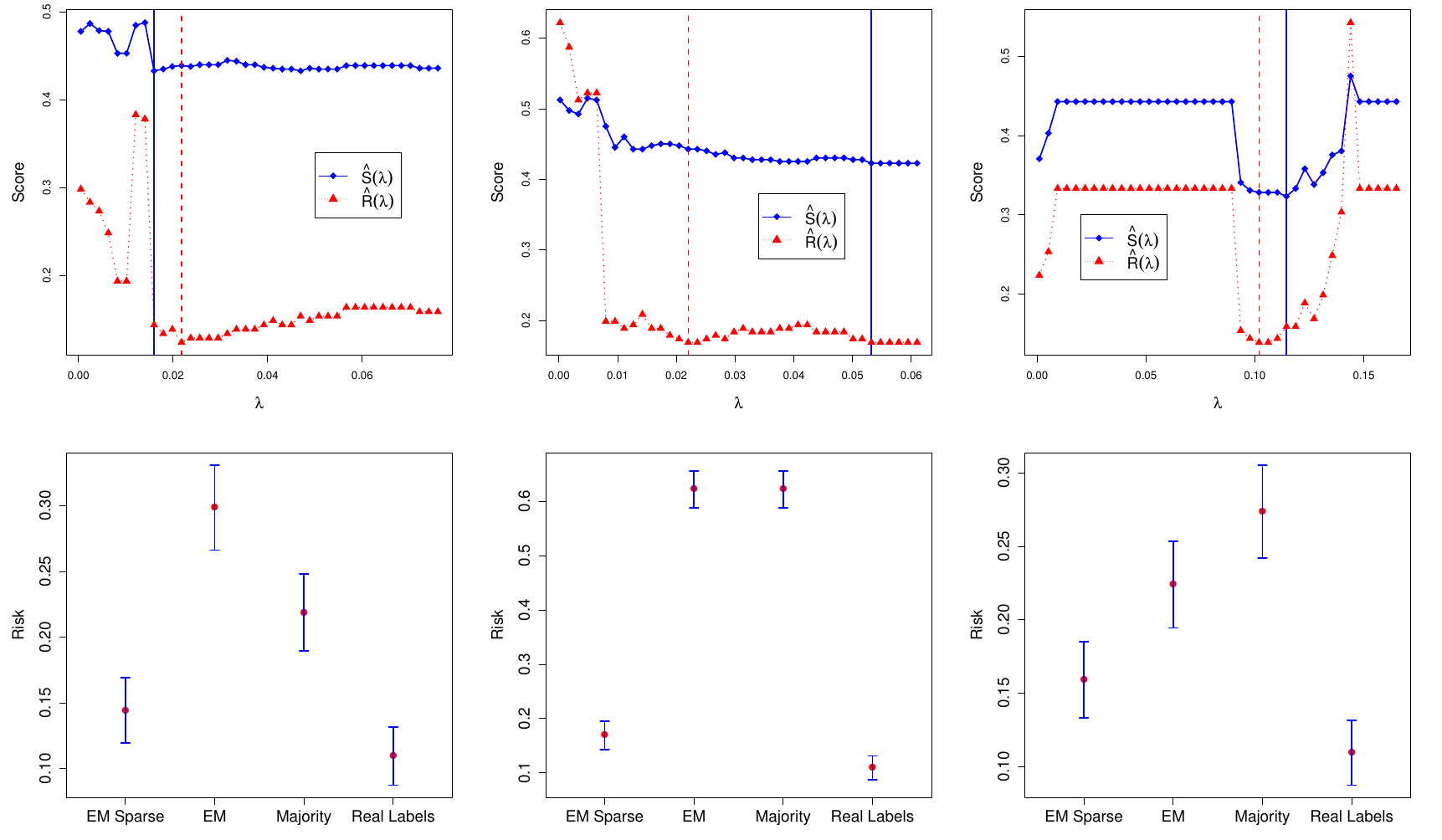}
  \caption{Top: Process of choosing $\lambda$ using $\widehat{S}(\lambda)$ for Ionosphere Data Set for models (i), (ii) and (iii) respectively. Vertical lines indicate
  where the minimum is attained for $\widehat{S}(\lambda)$ (solid) and for $\widehat{R}(\lambda)$ (dashed). Bottom: Estimated prediction errors for each dataset according to each model. }
  \label{ion}
\end{figure}

\begin{figure}[htpb!]
	\centering
      \includegraphics[scale=1]{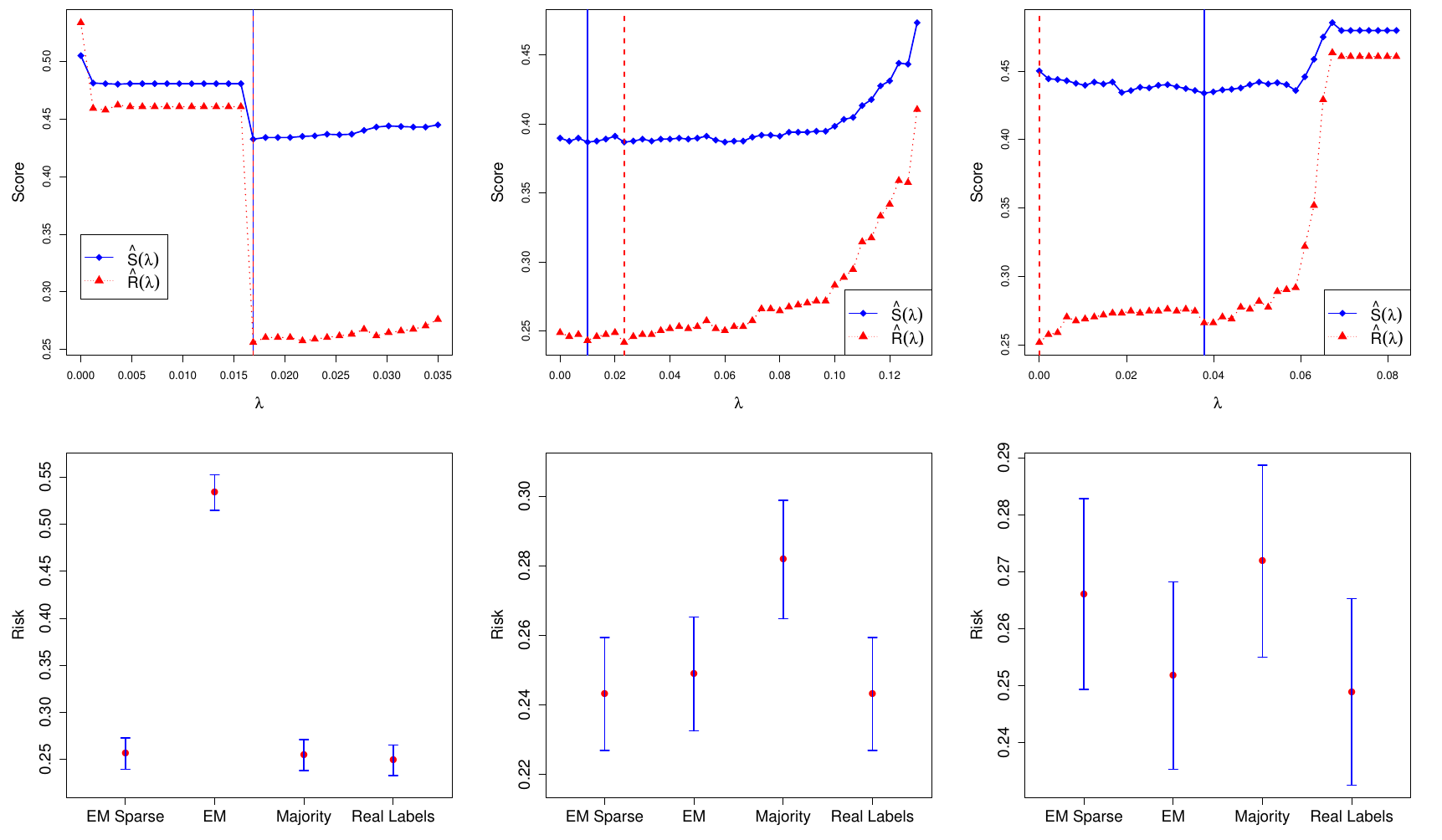}	
	\caption{Top: Process of choosing $\lambda$ using $\widehat{S}(\lambda)$ for Wine Data Set for models (i), (ii) and (iii) respectively. Vertical lines indicate
  where the minimum is attained for $\widehat{S}(\lambda)$ (solid) and for $\widehat{R}(\lambda)$ (dashed). Bottom: Estimated prediction errors for each dataset according to each model. }
	\label{wine}
\end{figure}

\begin{figure}[htpb!]
  \centering
  \includegraphics[scale=1]{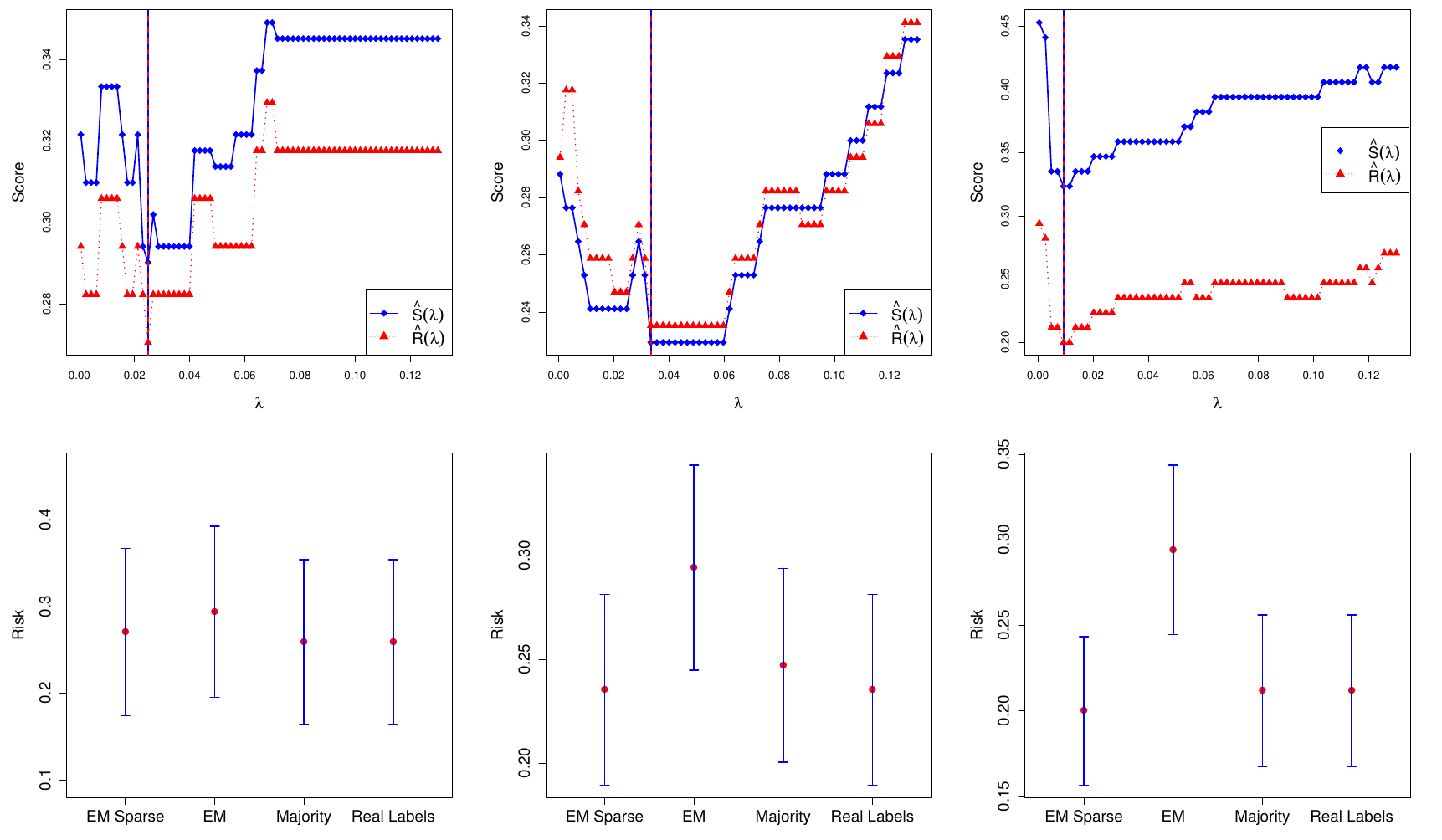}	
  \caption{Top: Process of choosing $\lambda$ using $\widehat{S}(\lambda)$ for Astronomy Data Set. Vertical lines indicate
  where the minimum is attained for $\widehat{S}(\lambda)$ (solid) and for $\widehat{R}(\lambda)$ (dashed). Bottom: Estimated prediction errors for each dataset according to each model. }
  \label{astroLambda}
\end{figure}

 \begin{table}[!h]
 	\begin{center}
 	\caption{Values of statistic $\widehat{S}$ for the experiments in \ref{applications}. Bold numbers stand for the minimizer of $\widehat{S}$, * indicates the minimizer of $\widehat{R}$.}			
 	\begin{tabular}{ccccc} \hline \label{SModels}
 		%\textbf{Dataset}&\textbf{Scenario}&\textbf{EM-Sparse}&\textbf{EM}&\textbf{Majority} \\ \hline \hline
 		& (i)&\textbf{0.434}* &0.485 & 0.446 \\ 
 		Simulated & (ii)&\textbf{0.459}* &0.510 &0.483  \\ 
		& (iii)&\textbf{0.462}* &0.463&0.477 \\\hline 
 		& (i)&\textbf{0.432}*&0.477 &0.435  \\ 
		Ionosphere & (ii)& \textbf{0.422}*&0.512 &0.544  \\ 
		& (iii)& \textbf{0.323}*&0.370 &0.398  \\\hline 
 		& (i)& 0.432&0.505 &\textbf{0.431}*  \\ 
		Wine & (ii)&\textbf{0.386}* &0.389 &0.399  \\ 
 		& (iii)& \textbf{0.433}&0.450*&0.451  \\\hline
		& (i)&0.290 &0.321 &\textbf{0.286}*  \\ 
		Astronomy & (ii)& \textbf{0.229}*&0.288 &0.241  \\ 
		& (iii)& \textbf{0.323}*&0.452 &0.335  \\\hline 
 	\end{tabular} 
 	\end{center}
 \end{table}

\section{Conclusions and Future Work}
\label{conclusions}

Dealing with noisy labels is a common problem. We show a way one can build classifiers that potentially have better performance than more traditional methods used when true labels are unavailable. The idea behind it is that sparsity is a good way to avoid overparametrizations and therefore creates classification schemes that may have better prediction errors. We also show how model selection can be performed, in particular how one can choose tuning parameters that induce sparsity. The method is based on the introduction of a surrogate function for the estimated risk. Both theoretical and empirical results indicate that the proposed method for model selection works under a fairly large class of problems. 

Even though in many situations latent variable models provide big improvements compared to majority vote 
(see \cite{Yan2010}and \cite{Raykar2010} for examples of such cases), we saw that in some cases the latter can perform better than the former. Two important reasons of why this happens are 1. The Expectation Maximization Algorithm is sensible to initialization and may converge to local minimums and 2. If the number of experts is large, majority vote can be accurate provided the voters are reasonably good. On the other hand, in such situations latent variable models have too many parameters to be estimated, and hence estimation is more difficult. This is a problem specially if the number of samples is small. However, sparsity can often diminish this problem, leading to estimators that may be better than the ones derived from majority vote procedures. A way to deal with this in practice is to use the proposed model selection technique to compare models built on majority vote labels and on models derived from latent variable models. Performing this procedure in our examples almost always led us 
to pick the model among \emph{EM Sparse, EM} and \emph{Majority} which had the smallest $\hat{R},$ which is the standard procedure when true labels are available.
%Moreover, we showed that by using the model selection technique for the tuning parameter one can detect some situations where the EM did not converge to the global minimum . 

Even though we focus on the approach of building models without first estimating the true labels of the data, the ideas of model selection presented are quite general. In fact, even when using the two step procedures (that first find the ``true'' labels 
either using majority vote or using fancier methods such as in \cite{crowd}, and then build classifiers based on the recovered labels), the technique proposed for  choosing between models is still valid. 
An advantage of the latent variables approach over two step procedures is that the first naturally allows partial information from experts to be incorporated when classifying new instances,
that is, one can easily calculate $p(z|x,y)$ for new data, even if not all experts observe the data point. 

On this paper, we used the same tuning parameter for both $\pmb{\gamma}'$s and $\pmb{\beta}'$s. As the roles of parameters are of different nature, in practice better performance can be achieved by using two different tuning parameters. This improvement comes at the expense of computational time.

Even though we only introduced sparsity for a specific model, the same arguments can be performed in different situations. For example, one could easily
create models in which $P(Y=1|Z=1) \neq P(Y=0|Z=0)$ by introducing new coefficients. It is also possible to use links different than the logit,
and also include dependencies that are not linear in the covariates that were observed. On can also introduce sparsity on approaches from the literature that
were already shown to be useful (e.g., \cite{Raykar2010,AAAIW125257}). 
Our model selection technique  helps choosing between these models.

There are also open questions regarding model selection through $\hat{S}$. Theorem \ref{modelSelectionMain} does not hold if any of the 
assumptions is removed. Necessary conditions for the consistency of minimizing $\hat{S}$ in 
model selection are unknown. It would also be useful to estimate $R(\lambda^{*})$. 
Theorem \ref{modelSelectionMain} shows that $S(\lambda^{*})$ is close to $(1-2\bar{\epsilon})R(\lambda^{*})+\bar{\epsilon}$. Hence, it might 
be possible to estimate $R(\lambda^{*})$ using $S(\lambda^{*})$ and an estimator for $\bar{\epsilon}$. Therefore it would be useful
to have consistent estimators of $\bar{\epsilon}$. Finally, 
we use $\hat{S}$ to find a consistent estimator under the $0$-$1$ loss. It remains unknown how to generalize this methodology for other loss functions. For example, in some classification of binary variables, the cost of error depends on the labels.

 \section{Acknowledgments}
 The authors are thankful for Peter E. Freeman, Georg M. Goerg, Ann B. Lee, Jennifer M. Lotz, Tiago Mendon\c{c}a Jeffrey A. Newman, Mauricio Sadinle and Larry Wasserman
 for the insightful comments. The authors would also like to thank the members of the CANDELS collaboration for providing the proprietary data and
annotations for the astronomy data.
 This work was partially supported by \emph{Conselho Nacional de Desenvolvimento Cient\'ifico e Tecnol\'ogico}.

\bibliographystyle{plainnat}
\bibliography{SADM}

\newpage

\appendix
{\huge\textbf{Appendix}}

\section{Maximization Step of EM}
\label{appendixMax}
\newcounter{equationNumber}

Here we give more details of how to find the maximum in Equation \ref{eqMax}. As we noted, we have two independent maximization problems,
\ref{eq1} and \ref{eq2}.

The first problem, \ref{eq1}, can be rewritten as 

$$ \arg \sup_{\gamma's,\alpha's} \mbox{ }\sum_{i=1}^{2dn} w(i)\log(\mu'_i),$$

where

$$
w_i = \left\{ \begin{array}{l}
 \mbox{ $i$-th element of the vectorization of the matrix $(d_{jl})_{1\leq j\leq n,1\leq l\leq d}$ for $1 \leq i\leq dn$} \\
  1-w_{i-dn} \mbox{ for $dn+1 \leq i\leq 2dn$}
       \end{array} \right.
$$

and

$$
\mu'_i = \left\{ \begin{array}{l}
\mbox{ $i$-th element of the vectorization of the matrix $(\mu_{jl})_{1\leq j\leq n,1\leq l\leq d}$ for $1 \leq i\leq dn$} \\
  1-\mu'_{i-dn} \mbox{ for $dn+1 \leq i\leq 2dn$}
       \end{array} \right.
$$

Here, 

$$ \mu_{ik}= \frac{\exp\{\alpha_k+\sum_j\gamma_j x_{ij}\}}{1+\exp\{\alpha_k+\sum_j\gamma_j x_{ij}\}} $$

This is just a Weighted L1-Regularized Logistic Regression, and can be solved using functions such as glmnet \citep{glmnet} in R. Alternatively, one can directly use algorithms such as Newton-Raphson.
Note that if the number of experts is larger than the number of features, using a sparse representation of the matrix can speed up the calculations.
The observations related to this maximization problem are

$$
\overbrace{1, \ldots, 1,}^{{}\mbox{dn times}} \overbrace{0, \ldots, 0}^{{}\mbox{dn times}}. 
$$

The second problem, \ref{eq2}, can be rewritten as 

$$ \arg \sup_\beta \mbox{ } \sum_{i=1}^{2n} w_i \log(\mu'_i)- \lambda\sum_{j=1}^k|\beta_j|,$$

where 

\begin{equation*}
 \begin{array}{cc}
w_i = \left\{ \begin{array}{rl}
 z_i &\mbox{ for $1 \leq i\leq n$} \\
  1-z_{i-n} &\mbox{ for $d+1 \leq i\leq 2n$}
       \end{array} \right. & \hspace{10mm}
\mu'_i = \left\{ \begin{array}{rl}
 \mu_i &\mbox{for $1 \leq i\leq n$} \\
  1-\mu_{i-n} &\mbox{ for $n+1 \leq i\leq 2n$}
       \end{array} \right.\\
\end{array}
\end{equation*}
\vspace{4mm}

This is again a Weighted L1-Regularized Logistic Regression.
The observations related to this maximization problem are

$$
\overbrace{1, \ldots, 1,}^{{}\mbox{n times}} \overbrace{0, \ldots, 0}^{{}\mbox{n times}}. 
$$

\section{Proofs}
\label{appendixProof}
 
 Our argument of why $\hat{S}$ is a good measure of performance can be decomposed in three steps. First, we show that the mean of $\hat{S}(\lambda)$ is close to $(1-2\bar{\epsilon})R(\lambda) + \bar{\epsilon}$. Next, we prove that, if $\Lambda$ is a VC-Class, then $\hat{S}(\lambda)$ approaches its mean uniformly on $\Lambda$. Finally, since $1-2\bar{\epsilon} > 0$ (Assumption \ref{assumptionError}), the minimizer of $(1-2\bar{\epsilon})R(\lambda) + \bar{\epsilon}$ is the same as the minimizer of $R(\lambda)$. From the three steps, we conclude that minimizing $\hat{S}(\lambda)$ approaches minimizing $R(\lambda)$. 

We use the following result found to relate $\hat{S}(\lambda)$ to the empirical risk:

\begin{Lemma}  
\label{scoreTransformation}
For all $\lambda \in \Lambda$ it holds that:

\[ \hat{S}(\lambda) = \frac{1}{n'}\sum_{i=1}^{n'}{\left(\frac{\sum_{j=1}^{d}{\left(1-2\I(Z_{i} \neq Y^{test}_{i,j})\right)}}{d}\right) \I(z^\lambda_{i} \neq Z_{i})} + \frac{1}{n'}\sum_{i=1}^{n'}{\frac{1}{d}\sum_{j=1}^{d}{\I(Z_{i} \neq Y^{test}_{i,j})}}  \]
\end{Lemma}

\begin{proof}
For any given $i$ and $j$, 
\begin{align*}
 &\I(z_{i}^{\lambda} \neq Y^{test}_{i,j}) = \I(z_{i}^{\lambda} \neq Y^{test}_{i,j},Z_{i} = Y^{test}_{i,j}) + \I(z_{i}^{\lambda} \neq Y^{test}_{i,j},Z_{i} \neq Y^{test}_{i,j}) = \\
  &= \I(z_{i}^{\lambda} \neq Z_{i} ,Z_{i} = Y^{test}_{i,j}) + \I(z_{i}^{\lambda} = Z_{i} ,Z_{i} \neq Y^{test}_{i,j}) =  \\
  &\I(z_{i}^{\lambda} \neq Z_{i}) (1-\I(Z_{i} \neq Y^{test}_{i,j})) + (1-\I(z_{i}^{\lambda} \neq Z_{i}))\I(Z_{i} \neq Y^{test}_{i,j}) = \\
  &= \I(z_{i}^{\lambda} \neq Z_{i}) (1-2\I(Z_{i} \neq Y^{test}_{i,j})) + \I(Z_{i} \neq Y^{test}_{i,j}) 
\end{align*}

$\hat{S}(\lambda)$ is obtained averaging $\I(z_{i}^{\lambda} \neq Y^{test}_{i,j})$ over $i$ and $j$. The right hand side of the lemma is obtained averaging $\I(z_{i}^{\lambda} \neq Z_{i}) (1-2\I(Z_{i} \neq Y^{test}_{i,j})) + \I(Z_{i} \neq Y^{test}_{i,j})$ over $i$ and $j$. Hence, the proof is complete.
\end{proof}

Observe that $\frac{1}{n'}\sum_{i=1}^{n'}{\frac{1}{d}\sum_{j=1}^{d}{\I(Z_{i} \neq Y^{test}_{i,j})}}$ is constant on $\lambda$. Hence, 

\[ \arg\min_{\lambda \in \Lambda} \hat{S}(\lambda) = \arg\min_{\lambda \in \Lambda} \frac{1}{n'}\sum_{i=1}^{n'}{\frac{\sum_{j=1}^{d}{\left(1-2\I(Z_{i} \neq Y^{test}_{i,j})\right)}}{d} \I(z^\lambda_{i} \neq Z_{i})} \]

The model which minimizes $\hat{S}(\lambda)$ minimizes a weighted average of $\I(z^\lambda_{i} \neq Z_{i})$. This is similar to performing model selection through empirical risk minimization, in which the model which minimizes the arithmetic mean of $\I(z^\lambda_{i} \neq Z_{i})$ is chosen.

\begin{Lemma} 
\label{inequality1}
Under assumption \ref{assumptionIndExpertExpert}, for all $\lambda \in \Lambda$ it holds that

\[ \left|E[\hat{S}(\lambda)] - (1-2\bar{\epsilon})R(\lambda) + \bar{\epsilon} \right|\leq \frac{\sigma_{z^{\lambda}}}{\sqrt{d}}, \]

where $\sigma^2_{z^{\lambda}} = VAR\left[\I(z_{i}^{\lambda} \neq Z_{i})\right]$.
\end{Lemma}
\begin{proof}
Let $W=\I(z_{i}^{\lambda} \neq Z_{i})$ and $V_{j}=1-2\I(Z_{i} \neq Y^{test}_{i,j})$. From Cauchy Schwartz inequality it follows that:

\[ \left|COV\left(W,\frac{\sum_j V_j}{d}\right)\right|\leq \sqrt{VAR\left[W\right]}\sqrt{VAR\left[\frac{\sum_j V_j}{d}\right]}=\sigma_{z^\lambda} \frac{1}{d}\sqrt{\sum_j VAR[V_j]}\leq \frac{\sigma_{z^\lambda}}{\sqrt{d}} \]

The conclusion follows from noticing that $\left|COV\left(W,\frac{\sum_j V_j}{d}\right)\right|$ is the left term
of the inequality presented.
\end{proof}

Hence, if we can conlude that $\hat{S}(\lambda)$ is close to its mean, since its mean is close to $(1-2\bar{\epsilon})R(\lambda) + \bar{\epsilon}$, we establish that minimizing $\hat{S}(\lambda)$ is close to minimizing $R(\lambda)$. The following result proves that $\hat{S}(\lambda)$ is close to its mean.

\begin{Lemma}  
\label{uniformConvergence}
If $\Lambda$ is a VC-Class then,

\[ P(\sup_{\lambda \in \Lambda}|\hat{S}(\lambda)-E[\hat{S}(\lambda)]| > \delta) \leq \left(\frac{D\sqrt{n'}\delta}{\sqrt{2VC(\Lambda)}}\right)^{2 VC(\Lambda)} e^{-2n'\delta^{2}} \]
\end{Lemma}
\begin{proof} 
Using Lemma \ref{scoreTransformation},

\[ \hat{S}(\lambda) = \frac{1}{n'}\sum_{i=1}^{n'}{\left(\frac{\sum_{j=1}^{d}{\left(1-2\I(Z_{i} \neq Y^{test}_{i,j})\right)}}{d}\right)} \I(z^\lambda_{i} \neq Z_{i}) + \frac{1}{n'}\sum_{i=1}^{n'}{\frac{1}{d}\sum_{j=1}^{d}{\I(Z_{i} \neq Y^{test}_{i,j})}} \]

Define $V_{i} = \frac{1}{d}\sum_{j=1}^{d}{\I(Z_{i} \neq Y^{test}_{i,j})}$ and $W_{i}^{\lambda} = \I(z^\lambda_{i} \neq Z_{i})$. Thus,

\[ \hat{S}(\lambda) = \frac{1}{n'}\left(\sum_{i=1}^{n'}{W_{i}^{\lambda}(1-2V_{i}) + V_{i}}\right) \]

We wish to prove that the central limit theorem holds uniformly on $S[\Lambda] = \{W^{\lambda}(1-2V) + V : \lambda \in \Lambda\}$. Let $N(\mathcal{F},\epsilon,L^{2}(Q))$ be the $L^{2}(Q)$ covering number of a class of functions, $\mathcal{F}$. Call $R[\Lambda] = \{W^{\lambda}: \lambda \in \Lambda\}$. Note that, since $|1-2V| \leq 1$, for every distribution $Q$, $N(S[\Lambda],\epsilon,L^{2}(Q)) \leq N(R[\Lambda],\epsilon,L^{2}(Q))$. Let $VC(\Lambda)$ be the VC-dimension of $\Lambda$. From \cite{Vaart}, $\sup_{Q}{N(R[\Lambda],\epsilon,L^{2}(Q))} \leq K \cdot VC(\Lambda) (4e)^{VC(\Lambda)} \left(\frac{1}{\epsilon}\right)^{2(VC(\Lambda)-1)}$. Hence, there exists a constant $D$, such that,

\[ P(\sup_{\lambda \in \Lambda}|\hat{S}(\lambda)-E[\hat{S}(\lambda)]| > \delta) \leq \left(\frac{D\sqrt{n'}\delta}{\sqrt{2VC(\Lambda)}}\right)^{2 VC(\lambda)} e^{-2n'\delta^{2}} \]
\end{proof}

Finally, putting together lemmas \ref{inequality1} and \ref{uniformConvergence}, we get Theorems \ref{modelSelectionMain} and \ref{modelSelectionMain2}.

\end{document}